\documentclass[a4paper,fleqn]{cas-sc}

\usepackage[numbers]{natbib}
\usepackage{algorithm}
\usepackage{algorithmic}
\newtheorem{proposition}{Proposition}
\newtheorem{proof}{Proof}

\newcommand{\STATEX}{\item[]}

\begin{document}
\let\WriteBookmarks\relax
\def\floatpagepagefraction{1}
\def\textpagefraction{.001}
\shorttitle{FALQON with Low Depth and Measurement}
\shortauthors{Zi-Wen Huang et~al.}

\title [mode = title]{Feedback-based quantum optimization with low depth and measurement}                      

\affiliation[1]{organization={State Key Laboratory of Networking and Switching Technology, Beijing University of Posts and Telecommunications},
                city={Beijing},
                postcode={100876},
                country={China}}

\affiliation[2]{organization={State Key Laboratory of Cryptology, P. O. Box 5159},
                city={Beijing},
                postcode={100878},
                country={China}}

\affiliation[3]{organization={China Electronic Technology Cyber Security Co., Ltd.},
                city={Chengdu},
                postcode={610041},
                country={China}}

\affiliation[4]{organization={National Key Laboratory of Security Communication, Institute of Southwestern Communication},
                city={Chengdu},
                postcode={610041},
                country={China}}

\affiliation[5]{organization={School of Cyberspace Security, Beijing University of Posts and Telecommunications},
                city={Beijing},
                postcode={100876},
                country={China}}

\author[1,2,5]{Zi-Wen Huang}[orcid=0009-0003-6009-5726]
\credit{Conceptualization, Methodology, Software, Formal analysis, Investigation, Data curation, Visualization, Writing – original draft, Writing – review \& editing}

\author[1,5]{Jia-Cheng Fan}[orcid=0009-0003-5758-019X]
\credit{Methodology, Software, Formal analysis, Writing – review \& editing}

\author[1,5]{Xiao-Hui Ni}[orcid=0009-0004-1713-1418]
\credit{Investigation, Validation, Visualization, Writing – review \& editing}

\author[1,5]{Su-Juan Qin}[orcid=0000-0002-6405-6711]
\credit{Data curation, Supervision, Writing – review \& editing}

\author[3,4]{Xiao-Kai Hou}[orcid=0000-0002-2869-3828]
\credit{Validation, Writing – review \& editing}

\author[4]{Wei Huang}[orcid=0000-0003-0586-090X]
\cormark[1]
\ead{huangwei096505@aliyun.com}
\credit{Resources, Project administration, Funding acquisition}

\author[4]{Bing-Jie Xu}
\credit{Resources, Supervision}

\author[1,2,5]{Fei Gao}[orcid=0000-0002-1546-4364]
\cormark[2]
\ead{gaof@bupt.edu.cn}
\credit{Conceptualization, Supervision, Project administration, Funding acquisition, Writing – review \& editing}

\cortext[cor1]{Corresponding author}
\cortext[cor2]{Principal corresponding author}

\begin{abstract}
Feedback-based ALgorithm for Quantum OptimizatioN (FALQON) is a hybrid quantum-classical algorithm for solving combinatorial optimization problems, which circumvents classical parameter optimization but requires a deep quantum circuit.
To reduce circuit depth, Arai \textit{et al.} proposed second-order FALQON (SO-FALQON), achieving the best depth reduction among existing approaches.
However, SO-FALQON brings a 2.3 times per-step measurement overhead, as it needs to calculate an additional second-order control coefficient.
In this paper, inspired by the Backtracking Line Search (BLS) theory, we propose another method called BLS-FALQON, which not only reduces circuit depth to a comparable extent, but also achieves fewer measurements than SO-FALQON.
Numerical simulations on max-cut problem with 8 to 20 vertices demonstrate that BLS-FALQON reduces the total measurement count by 37.7\% compared to SO-FALQON, while maintaining a comparable circuit depth.
Furthermore, we conduct real quantum hardware experiments on the \textit{Tianyan-176} quantum computer, which uses the \textit{zuchongzhi2} superconducting quantum processor, confirming that BLS-FALQON remains effective under real quantum hardware conditions.
\end{abstract}




\maketitle

\section{Introduction}

Feedback-based algorithm for quantum optimization (FALQON)~\citep{magann2022feedback,magann2022lyapunov} is a hybrid quantum-classical algorithm inspired by quantum Lyapunov control~\citep{kosloff1992excitation,ohtsuki1998application,j1999laser,grivopoulos2003lyapunov,engel2009local}. Unlike conventional variational quantum algorithms (VQAs)~\citep{peruzzo2014variational,farhi2014quantum,liu2021variational,bravo2023variational,cerezo2021variational}, it circumvents the need for any classical optimization effort, thereby avoiding the inherent difficulty of training variational parameters~\citep{bittel2021training}. However, FALQON requires multiple iterative steps where each iteration appends an additional layer of Hamiltonian evolution to the quantum state, which inherently results in a deep quantum circuit. Minimizing quantum resource consumption, such as circuit depth and measurement count, remains a critical challenge across diverse practical quantum applications, which encompass combinatorial optimization~\citep{zhou2020quantum,harrigan2021quantum,kruger2025out} and quantum cryptanalysis~\citep{regev2025efficient,ragavan2024space,xiao2026offline}. Consequently, the deep circuit required by FALQON severely hinders its practical applications on noisy intermediate-scale quantum (NISQ) devices~\citep{preskill2018quantum,li2025efficient,su2025topology}.

The circuit depth of FALQON is governed by its discrete time step $\Delta t$, because a larger $\Delta t$ always converges with fewer layers for a fixed target solution quality. However, increasing $\Delta t$ amplifies the Trotter decomposition errors introduced at each layer. These errors can misdirect the feedback signal and cause the problem Hamiltonian's expectation value to oscillate rather than decrease monotonically, which ultimately prevents convergence. Consequently, the original FALQON is restricted to a small $\Delta t$, resulting in prohibitive circuit depth~\cite{magann2022lyapunov,arai2025scalable}.

To reduce circuit depth, Arai \textit{et al.} proposed a variant of FALQON, named second-order FALQON (SO-FALQON)~\cite{arai2025scalable}. SO-FALQON improves the feedback law of the original ones which depends on the first-order Taylor expansion of the cost function difference per step to a second-order version. This improvement allows SO-FALQON to use a larger discrete time step $\Delta t$ for quantum evolution while avoiding the oscillations caused by Trotter decomposition errors that could lead to non-convergence. Therefore, SO-FALQON can reduce circuit depth compared to FALQON~\cite{arai2025scalable}, which achieves a larger depth reduction than other approaches in this metric~\cite{malla2024feedback,tang2025nonvariational,brady2025feedback,rattighieri2025accelerating}. However, the second-order feedback law needs to estimate two additional complex Hamiltonians at every layer. This estimation increases the per-step measurement overhead to approximately 2.3 times that of the original FALQON.

In this paper, inspired by the backtracking line search (BLS) theory~\cite{armijo1966minimization,bertsekas1997nonlinear,nocedal2006numerical}, we propose BLS-FALQON to achieve both low circuit depth and low measurement count. Unlike SO-FALQON, which introduces additional computational overhead by calculating second-order control coefficient $\beta^{\text{(SO)}}$ to reduce the circuit depth of FALQON, BLS-FALQON successfully maintains a large $\Delta t$ without oscillations and completely avoids the estimation of complex Hamiltonians. Our proposed approach incorporates a modified BLS strategy with a negative scaling factor to heuristically explore the original first-order control coefficient $\beta^{\text{(FO)}}$. Specifically, BLS-FALQON adjusts $\beta^{\text{(FO)}}$ solely when the expectation value of the problem Hamiltonian $\langle H_p\rangle$ fails to decrease monotonically, which serves as a simplified Armijo condition~\cite{armijo1966minimization} to guarantee convergence stability. In such cases, we iteratively multiply $\beta^{\text{(FO)}}$ by a negative factor $\tau\in(-1,0)$ until $\langle H_p\rangle$ successfully decreases. Consequently, this operation allows our algorithm to achieve both low circuit depth by maintaining a large $\Delta t$ without oscillations and low measurement overhead by circumventing the estimation of complex Hamiltonians required in the second-order variant.

Our algorithm is evaluated on max-cut problem, which is an NP-hard problem that serves as a common benchmark for most hybrid quantum-classical algorithms in combinatorial optimization~\cite{ni2024multilevel,ni2025progressive,ni2026adaptive,huang2026iterative}. We compare BLS-FALQON, SO-FALQON, and FALQON across two metrics using both numerical simulations and real quantum hardware experiments.

For numerical simulations, we use unweighted 3-regular graphs with 8 to 20 vertices, randomly generating 50 non-isomorphic instances per size for a total of 274 instances (except for sizes 8 and 10, which have only 5 and 19 non-isomorphic instances, respectively). Numerical results show that BLS-FALQON reduces the total circuit depth by approximately 90.9\% and measurement count by 88.1\% relative to FALQON. Furthermore, our proposed algorithm reduces the total measurement count by 37.7\% relative to SO-FALQON, while it maintains a comparable circuit depth to reach the target approximation ratio benchmark. This indicates that our algorithm achieves a significant reduction in both circuit depth and measurement resources.

For real quantum hardware experiments, we execute all three algorithms on the \textit{Tianyan-176} quantum computer, which uses the \textit{zuchongzhi2} superconducting quantum processor~\cite{wu2021strong,zhu2022quantum,cao2023generation,qian2025programmable}. Due to the limited coherence time of the processor, experiments are restricted to unweighted 3-regular graph instances at $n\in\{4,6,8\}$ with a maximum of $L=10$ circuit layers. Following the calculation method in Ref.~\cite{guerreschi2019qaoa}, we evaluate the practical execution time of our algorithm against FALQON and its second order variant. The results demonstrate that BLS-FALQON effectively reduces the estimated average hardware execution time by 43\% relative to SO-FALQON, while remaining effective and stable under realistic hardware noise. This confirms that our algorithm is compatible with current NISQ devices in the context of the small-scale max-cut problem.

The rest of this paper is organized as follows. In Sec.~\ref{sec:pre} we briefly review FALQON and its second-order variant. In Sec.~\ref{sec:bls} we present our proposed BLS-FALQON and detail its key components. Sec.~\ref{sec:simu} reports the results of our numerical simulations using the Qiskit framework. Sec.~\ref{sec:real} reports the results of our experiments on real quantum hardware. Finally, we provide our conclusions in Sec.~\ref{sec:conclu}.

\section{\label{sec:pre}Preliminaries}

In this section, we review some relevant preliminary knowledge to help readers better understand our work, including quantum Lyapunov control, FALQON, and its second-order variant.

\subsection{Quantum Lyapunov control}

FALQON is inspired by quantum Lyapunov control (QLC)~\cite{kosloff1992excitation,ohtsuki1998application,j1999laser,grivopoulos2003lyapunov,engel2009local},  a continuous-time control method. QLC drives a closed quantum system toward the ground state of a target problem Hamiltonian by applying a time-dependent Hamiltonian. To illustrate the principle of QLC, we consider a time-dependent Hamiltonian:
\begin{equation}
    H(t)=H_p+\beta(t)H_d.
\end{equation}
Here, $H_p$ is the problem Hamiltonian, $H_d$ is a fixed driving Hamiltonian that does not commute with $H_p$, and $\beta(t)$ is a real-valued control function. The dynamics of a quantum state $|\psi(t)\rangle$ follow the Schr\"{o}dinger equation:
\begin{equation}
    \label{eq:Schrodinger}
    i\hbar\frac{d}{dt}|\psi(t)\rangle=H(t)|\psi(t)\rangle.
\end{equation}
The core objective of QLC is to force the expectation value of the problem Hamiltonian, denoted as
\begin{equation}
    \langle H_p \rangle\equiv \langle\psi(t)|H_p|\psi(t)\rangle,
\end{equation}
to decrease monotonically over time.

Taking the time derivative of $\langle H_p \rangle$ via the product rule and noting that the partial derivative $\partial H_p / \partial t$ vanishes for a time-independent problem Hamiltonian, we substitute Eq.~\ref{eq:Schrodinger} into the remaining terms to obtain:
\begin{equation}
    \frac{d}{dt}\langle H_p \rangle = \frac{i}{\hbar}\langle\psi(t)|[H(t),H_p]|\psi(t)\rangle.
\end{equation}
Since $H_d$ does not commute with $H_p$ and $[H_p,H_p]=0$, the time derivative of $\langle H_p\rangle$ is finally given by:
\begin{equation}
    \frac{d}{dt}\langle H_p \rangle = \beta(t)\frac{i}{\hbar}\langle\psi(t)|[H_d,H_p]|\psi(t)\rangle=\beta(t) A(t),
\end{equation}
where $A(t)=\frac{i}{\hbar}\langle\psi(t)|[H_d,H_p]|\psi(t)\rangle$.

To guarantee a monotonic decrease requirement $\frac{d}{dt}\langle H_p \rangle\leq0$ for all $t\geq0$, QLC sets the control function $\beta(t)$ to ensure that the product $\beta(t) A(t)$ remains non-positive at all times. By adhering to this continuous control law, the system is steered toward the lowest energy eigenstate of $\langle H_p\rangle$. However, QLC cannot be directly implemented on gate-based quantum computers, because these devices cannot natively execute continuous-time evolution.

\subsection{FALQON}

To implement the continuous-time QLC on gate-based quantum computers, FALQON translates the time evolution into a layered quantum circuit via Trotter decomposition using a discrete time step $\Delta t$~\cite{magann2022feedback,magann2022lyapunov}. It begins with an empty quantum circuit and an input uniform superposition quantum state $|\psi_0\rangle={|+\rangle}^{\otimes n}$. In the $k$-th step, it appends a new layer consisting of the problem unitary $U_p$ and the driving unitary $U_d(\beta^{\text{(FO)}}_k)$ to the end of the circuit (to distinguish $\beta$ from second-order variant in the following text, we will use $\beta^{\text{(FO)}}$ here, meaning first-order). It defines these unitary operators as (setting $\hbar=1$)
\begin{equation}
    U_p=e^{-iH_p\Delta t},
\end{equation}
and
\begin{equation}
    U_d(\beta^{\text{(FO)}}_k)=e^{-i\beta^{\text{(FO)}}_kH_d\Delta t}.
\end{equation}
The output quantum state can be represented by
\begin{equation}
    |\psi_k\rangle=U_d(\beta^{\text{(FO)}}_k)U_p|\psi_{k-1}\rangle.
\end{equation}

After preparing $|\psi_k\rangle$ at the $k$-th step, FALQON computes the discrete counterpart of $A(t)$ in QLC by estimating an expectation value, which is defined as
\begin{equation}
    \label{eq:Ak}
    A_k=i\langle\psi_k|[H_d,H_p]|\psi_k\rangle.
\end{equation}
The algorithm then determines the control coefficient for the next layer as (initializing $\beta^{\text{(FO)}}_1=0$)
\begin{equation}
    \label{eq:bkAK}
    \beta^{\text{(FO)}}_{k+1}=-A_k.
\end{equation}
This iterative procedure of utilizing quantum measurement results to classically assign the control coefficient for the next layer constitutes the feedback law of FALQON as outlined in Algorithm~\ref{alg:falqon}.

\begin{algorithm}[H]
\caption{FALQON}
\label{alg:falqon}
\begin{algorithmic}[1]
    \STATEX \textbf{Input:} Problem Hamiltonian $H_p$; driving Hamiltonian $H_d$; time step $\Delta t$; total layers $L$.
    \STATEX \textbf{Output:} Bit string $\psi_L$.
    \STATE Initialize $k = 1$; $\beta^{\text{(FO)}}_1 = 0$; $|\psi_0\rangle = |+\rangle^{\otimes n}$.
    \WHILE{$k \le L$}
        \STATE Prepare quantum state $|\psi_k\rangle = U_d(\beta^{\text{(FO)}}_k)U_p|\psi_{k-1}\rangle$.
        \STATE Measure to estimate expectation $A_k = i\langle\psi_k|[H_d, H_p]|\psi_k\rangle$.
        \STATE Set the next control coefficient $\beta^{\text{(FO)}}_{k+1} = -A_k$.
        \STATE Proceed to the next layer $k \leftarrow k + 1$.
    \ENDWHILE
\RETURN Z basis measurements on $|\psi_L\rangle$.
\end{algorithmic}
\end{algorithm}

However, the Trotter decomposition errors force FALQON to use an extremely small step size to guarantee a monotonically decreasing $\langle H_p\rangle$, which inevitably leads to prohibitive circuit depth. Although the authors empirically chose a larger time step, denoted as $\Delta t_c$, the required quantum circuit depth is still very deep~\cite{magann2022lyapunov,arai2025scalable}.

\subsection{Second-order FALQON}

To reduce circuit depth, second-order FALQON extends the original feedback law~\cite{arai2025scalable}. It expands the difference in $\langle H_p\rangle$ between consecutive steps into a Taylor series with respect to $\Delta t$. Unlike the original feedback law that solely forces the first-order term to decrease, the extended feedback law determines the control coefficient by ensuring the sum of the first-order and second-order terms is strictly negative. This incorporation of second-order terms allows a monotonic decrease in $\langle H_p\rangle$ for a larger $\Delta t$ without causing oscillations.

To formulate this extended feedback law, the authors define two additional expectation values $B_k$ and $C_k$ as
\begin{equation}
    \label{eq:Bk}
    B_k=\langle\psi_k|\frac{1}{2}[[H_d,H_p],H_d]|\psi_k\rangle,
\end{equation}
and
\begin{equation}
    \label{eq:Ck}
    C_k=\langle\psi_k|[[H_d,H_p],H_p]|\psi_k\rangle.
\end{equation}
Based on these expectation values, the second-order control coefficient $\beta^{\text{(SO)}}_{k+1}$ is derived as (initializing $\beta^{\text{(SO)}}_1=0$)
\begin{equation}
    \label{eq:SO}
    \beta^{\text{(SO)}}_{k+1}=-\frac{A_k+\Delta tC_k}{2\Delta tB_k}.
\end{equation}
Notably, the authors empirically select whichever coefficient between $\beta^{\text{(FO)}}_{k+1}$ and $\beta^{\text{(SO)}}_{k+1}$ has the smaller absolute value to serve as the actual control coefficient for the next layer. Numerical simulations demonstrate that this selection strategy achieves faster convergence than relying solely on the second-order coefficient~\cite{arai2025scalable}. The steps of SO-FALQON are shown in Algorithm~\ref{alg:so_falqon}.

Calculating $\beta^{\text{(SO)}}_{k+1}$ requires measuring two additional expectation values $B_k$ and $C_k$ at each circuit layer, which consist of sums of multiple Pauli products. To minimize the measurement overhead, the authors group them into the minimal number of qubit-wise commuting subsets. This grouping task, which can be equivalently mapped to a graph coloring problem, is easily solved by the Welsh-Powell algorithm~\cite{welsh1967upper} approximately. Although Welsh-Powell algorithm reveals that the second-order coefficient requires 2.3 times more measurement bases per step on average, the reduced circuit depth resulting from a larger step size ensures that the measurement overhead to achieve the target solution quality remains lower than that of the FALQON.

\begin{algorithm}[H]
\caption{SO-FALQON}
\label{alg:so_falqon}
\begin{algorithmic}[1]
    \STATEX \textbf{Input:} Problem Hamiltonian $H_p$; driving Hamiltonian $H_d$; time step $\Delta t$; total layers $L$.
    \STATEX \textbf{Output:} Bit string $\psi_L$.
    \STATE Initialize $k = 1$; $\beta^{\text{(SO)}}_1 = 0$; $|\psi_0\rangle = |+\rangle^{\otimes n}$.
    \WHILE{$k \le L$}
        \STATE Prepare quantum state $|\psi_k\rangle = U_d(\beta^{\text{(SO)}}_k)U_p|\psi_{k-1}\rangle$.
        \STATE Measure to estimate expectations $A_k = i\langle\psi_k|[H_d, H_p]|\psi_k\rangle$, $B_k=\langle\psi_k|\frac{1}{2}[[H_d,H_p],H_d]|\psi_k\rangle$, and $C_k=\langle\psi_k|[[H_d,H_p],H_p]|\psi_k\rangle$.
        \IF{$|A_k|<|\frac{A_k+\Delta tC_k}{2\Delta tB_k}|$}
            \STATE Set the next control coefficient $\beta^{\text{(SO)}}_{k+1} = -A_k$.
        \ELSE
            \STATE Set the next control coefficient $\beta^{\text{(SO)}}_{k+1} = -\frac{A_k+\Delta tC_k}{2\Delta tB_k}$.
        \ENDIF
        \STATE Proceed to the next layer $k \leftarrow k + 1$.
    \ENDWHILE
\RETURN Z basis measurements on $|\psi_L\rangle$.
\end{algorithmic}
\end{algorithm}

\section{\label{sec:bls}Methodology}

In this section, we propose BLS-FALQON, a heuristic variant designed to reduce the measurement overhead relative to second-order variant while preserving the circuit depth reduction achieved by large time steps. Our approach discards the estimation of $B_k$ and $C_k$. 
We integrate the concept of backtracking line search to heuristically adjust the original first-order control coefficient $\beta^{\text{(FO)}}$. The algorithm dynamically scales $\beta^{\text{(FO)}}$ upon detecting an increase of problem Hamiltonian's expectation value. This strategy stabilizes the quantum evolution at a large step size while restricting the measurement requirements to $H_p$ and $[H_d,H_p]$.

\subsection{Backtracking Line Search FALQON}

In the following, we detail the execution procedure of BLS-FALQON. It begins with the same as FALQON. In addition to calculating $A_k$, BLS-FALQON also computes the problem Hamiltonian's expectation value at the $k$-th step, which is defined as
\begin{equation}
    \label{eq:Ek}
    E_k=\langle\psi_k|H_p|\psi_k\rangle.
\end{equation}
If $E_k\le E_{k-1}$, our approach assigns Eq.~\ref{eq:bkAK} and proceeds to the next layer. Conversely, the following backtracking mechanism is triggered. Physically, Trotter decomposition errors introduced by a large discrete time step may misdirect the control direction or overestimate the control strength. To counteract these potential misdirection, our approach applies a scaling factor to the current control coefficient. We write this iterative update rule as
\begin{equation}
    \beta^{\text{(FO)}}_{k}\leftarrow\tau\beta^{\text{(FO)}}_{k},
\end{equation}
where $\tau\in(-1,0)$. The negative sign reverses the potentially misdirected control operation, and the fractional magnitude concurrently reduces the control strength to suppress severe oscillations. BLS-FALQON repeatedly scales the parameter and re-calculates $E_k$ until $E_k\le E_{k-1}$. The complete algorithmic structure is summarized in Algorithm~\ref{alg:bls_falqon}.

We note that Algorithm~\ref{alg:bls_falqon} restructures the backtracking detection relative to the description above. Rather than measuring $E_k$ within step $k$ and conditionally measuring $A_k$ afterward, the inner while loop at the beginning of each step checks whether $E_{k-1}$ has increased relative to $E_{k-2}$ and corrects $\beta^{\text{(FO)}}_{k-1}$ if necessary. This restructuring allows $E_k$ and $A_k$ to be estimated simultaneously by using Welsh-Powell algorithm within each normal step, rather than separately.

\begin{algorithm}[H]
\caption{BLS-FALQON}
\label{alg:bls_falqon}
\begin{algorithmic}[1]
    \STATEX \textbf{Input:} Problem Hamiltonian $H_p$; driving Hamiltonian $H_d$; time step $\Delta t$; total layers $L$; scaling factor $\tau$.
    \STATEX \textbf{Output:} Bit string $\psi_L$
    \STATE Initialize $k=1$; $\beta^{\text{(FO)}}_1=0$; $E_{-1}=E_0=+\infty$; $|\psi_0\rangle=|+\rangle^{\otimes n}$.
    \WHILE{$k\le L$}
        \WHILE{$E_{k-1}>E_{k-2}$}
            \STATE Update $\beta^{\text{(FO)}}_{k-1} \leftarrow \tau \beta^{\text{(FO)}}_{k-1}$.
            \STATE Update $|\psi_{k-1}\rangle = U_d(\beta^{\text{(FO)}}_{k-1})U_p|\psi_{k-2}\rangle$.
            \STATE Measure to estimate expectation $E_{k-1} = \langle\psi_{k-1}|H_p|\psi_{k-1}\rangle$.
        \ENDWHILE
        \STATE Prepare quantum state $|\psi_k\rangle = U_d(\beta^{\text{(FO)}}_k)U_p|\psi_{k-1}\rangle$.
        \STATE Measure to estimate expectations $A_k = i\langle\psi_k|[H_d, H_p]|\psi_k\rangle$ and $E_k = \langle\psi_k|H_p|\psi_k\rangle$.
        \STATE Set the next control coefficient $\beta^{\text{(FO)}}_{k+1} = -A_k$.
        \STATE Proceed to the next layer $k \leftarrow k + 1$.
    \ENDWHILE
\RETURN Z basis measurements on $|\psi_L\rangle$.
\end{algorithmic}
\end{algorithm}

Although BLS-FALQON introduces additional measurements of the problem Hamiltonian during each normal step and backtracking, estimating $E_k$ is significantly less measurement cost than $B_k$ and $C_k$ required by SO-FALQON. It is worth mentioning that for combinatorial optimization problems whose $H_p$ consists of 2-local $ZZ$ Pauli strings, such as the max-cut problem, all of which can be absorbed into the existing qubit-wise commuting subsets of the $A_k$ grouping. For example, jointly estimating $E_k$ and $A_k$ incurs no additional measurement bases for $n\ge12$ 3-regular max-cut problem, as confirmed in Table~\ref{tab:nbasis}. Furthermore, by enabling a large time step, our approach circumvents the prohibitive circuit depth required by FALQON. Consequently, both the measurement overhead and the circuit depth are minimized compared to FALQON. We will verify this advantage in Sec.~\ref{sec:simu}.

\subsection{Theoretical Analysis}

After analyzing why BLS-FALQON can jointly reduce circuit depth and total measurement bases count, to further justify this heuristic approach, we then analyze its theoretical viability from two distinct perspectives.

\subsubsection{Analysis from error bounds}

We prove that the backtracking mechanism terminates within a finite number of steps. We adopt the same discrete time step as SO-FALQON, for which Ref.~\cite{arai2025scalable} has numerically verified that higher-order terms are negligible compared to the first- and second-order terms. Under this assumption, by applying the BCH formula~\cite{rossmann2006lie,brian2003lie}, the single-step energy difference $\Delta E(\beta)\equiv E_k(\beta)-E_{k-1}$ can be expanded as
\begin{equation}
    \label{eq:DeltaE}
    \Delta E(\beta)=\beta\Delta tA_k+\beta^2\Delta t^2B_k+\beta\Delta t^2C_k+\mathcal{O}(\Delta t^3),
\end{equation}
where $A_k$, $B_k$ and $C_k$ are defined in Eqs.~\ref{eq:Ak},~\ref{eq:Bk} and~\ref{eq:Ck}. We define $\beta_c$ as the largest $|\beta|$ satisfying
\begin{equation}
    |\beta\Delta tA_k|>|R(\beta,\Delta t)|,
\end{equation}
where $R$ collects all terms of order $\mathcal{O}(\Delta t^2)$ and higher. Therefore, whenever $\beta A_k<0$ and $|\beta|\le\beta_c$, we have $\Delta E(\beta)<0$.

Due to Trotter decomposition errors introduced by a larger $\Delta t$, the measured $A_k$ may carry the wrong sign relative to the true energy gradient, so Eq.~\ref{eq:bkAK} may itself point in the wrong descent direction. The following proposition shows that the alternating sign of $\tau\in(-1,0)$ guarantees that, regardless of this potential misdirection, the backtracking mechanism always contains a valid descent step.

\begin{proposition}
    Let $A_k\ne0$ and $\tau\in(-1,0)$. The backtracking sequence $\beta^{(m)} = \tau^m\beta^{\text{(FO)}}$ with $\beta^{\text{(FO)}}=-A_k$ satisfies $\Delta E(\beta^{(m^*)})<0$ for some finite index $m^*$ bounded by 
    \begin{equation}
        M=\left\lceil\log_{|\tau|}\frac{\beta_c}{|A_k|}\right\rceil+1.
    \end{equation}
\end{proposition}

\begin{proof}
    We consider following two cases.

    \textbf{Case 1: $\beta^{(FO)}$ is in the correct descent direction}. Since $|\tau|<1$, the magnitudes $|\beta^{(m)}|=|\tau|^m|A_k|$ decay geometrically. At the first even index $m^*\ge m_0=\lceil\log_{|\tau|}(\beta_c/|A_k|)\rceil$, we have $|\beta^{(m^*)}|\le\beta_c$ and the same sign as $\beta^{(FO)}$, so $\Delta E(\beta^{(m^*)})<0$.

    \textbf{Case 2: $\beta^{(FO)}$ is in the wrong descent direction}. Then $\tau\beta^{(FO)}$ (the $m=1$ term) has the opposite sign from $\beta^{(FO)}$, restoring the correct descent direction. The odd-indexed subsequence $\{\beta^{(1)},\beta^{(3)},\dots\}$ therefore all point in the correct direction. By the same magnitude decay argument, there exists a first odd index $m^*\le M$ such that $|\beta^{(m^*)}|\le\beta_c$, guaranteeing $\Delta E(\beta^{(m^*)}) < 0$.

    In both cases $m^*\le M$, completing the proof.
\end{proof}

Consequently, the negative $\tau$ plays two roles. The alternating sign ensures that every other trial corrects any misdirected control direction, while the decaying magnitude $|\tau|^m$ drives the sequence into the convergence domain $[-\beta_c, \beta_c]$. Note that evaluating Eq.~\ref{eq:Ek} during each backtracking can most likely share measurement bases with the $A_k$ Pauli grouping, so the average additional measurement overhead per circuit layer is at most $M$ evaluations.

\subsubsection{Analysis from quadratic landscapes}

We then analyze the backtracking mechanism within the quadratic approximation framework of SO-FALQON. Eq.~\ref{eq:DeltaE} can be written in quadratic form in $\beta$:
\begin{equation}
    \Delta E(\beta)\approx\beta\Delta t(A_k+\Delta tC_k)+\beta^2\Delta t^2B_k.
\end{equation}
Since $B_k>0$ always holds in our approach (consistent with Ref.~\cite{arai2025scalable}), this is an upward-opening parabola in $\beta$ with roots at $\beta=0$ and $\beta=-(A_k+\Delta tC_k)/(\Delta tB_k)$. The parabola is negative between the two roots and positive outside them. Therefore, $\Delta E(\beta)<0$ if and only if $\beta$ is between $0$ and $-(A_k+\Delta tC_k)/(\Delta tB_k)$, which requires simultaneously that $\beta$ and $(A_k+\Delta tC_k)$ have opposite signs and that $|\beta|<|A_k+\Delta tC_k|/(\Delta tB_k)$. We refer to these two requirements together as the descent condition.

The second-order control coefficient $\beta^{\text{(SO)}}$ (Eq.~\ref{eq:SO}) always satisfies the descent condition. By contrast, the first-order control coefficient $\beta^{\text{(FO)}}=-A_k$ may violate the descent condition.

\begin{proposition}
    \label{prop:2}
    When the descent condition is violated at $\beta^{\text{(FO)}}$, the backtracking sequence $\beta^{(m)} = \tau^m\beta^{\text{(FO)}}$ with $\tau\in(-1,0)$ satisfies the descent condition at $\beta^{(m^*)}$ for some
    \begin{equation}
        \label{eq:prop2}
        m^*\le\lfloor\log_{|\tau|}\frac{|A_k+\Delta tC_k|}{|\Delta tA_kB_k|}\rfloor + 2.
    \end{equation}
\end{proposition}

\begin{proof}
    We address the sign condition and the magnitude condition separately.

    \textbf{Sign condition.} The descent condition requires $\beta$ and $-(A_k + \Delta t C_k)$ to have the same sign. Since $\tau < 0$, consecutive elements of the sequence alternate in sign. Therefore, either all even-indexed or all odd-indexed terms satisfy the sign condition.

    \textbf{Magnitude condition.} The descent condition further requires
    \begin{equation}
        |\beta^{(m)}|<\frac{|A_k+\Delta tC_k|}{\Delta tB_k}.
    \end{equation}
    Substituting $|\beta^{(m)}|=|\tau|^m|A_k|$, the above equation becomes
    \begin{equation}
        |\tau|^m|A_k|<\frac{|A_k+\Delta tC_k|}{\Delta tB_k}.
    \end{equation}
    Taking $\log_{|\tau|}$ of both sides (and reversing the inequality since $\log_{|\tau|}$ is decreasing) gives
    \begin{equation}
        m>\log_{|\tau|}\frac{|A_k+\Delta tC_k|}{|\Delta tA_kB_k|}.
    \end{equation}
    Hence the magnitude condition is satisfied for all
    \begin{equation}
        m\ge m_0\equiv\lfloor\log_{|\tau|}\frac{|A_k+\Delta tC_k|}{|\Delta tA_kB_k|}\rfloor+1.
    \end{equation}

    The productive subsequence covers either all even or all odd indices, so the first productive index $m^*$ satisfying $m^*\ge m_0$ is at most $m_0+1$, which establishes Eq.~\ref{eq:prop2}. At this index both conditions hold simultaneously, so $\Delta E(\beta^{(m^*)})<0$.
\end{proof}

Compared with SO-FALQON, which directly computes the optimal vertex $\beta^{\text{(SO)}}$, BLS-FALQON finds a heuristic coefficient $\beta^{(m^*)}$ that satisfies the descent condition but may deviate from the optimum. This may slightly reduce the energy drop per step and marginally increase the total circuit depth. However, by entirely eliminating the measurement of $B_k$ and $C_k$, the per-step measurement savings outweigh this overhead, as we verify numerically in the next section.

\section{\label{sec:simu}Numerical Simulations}

In this section, we evaluate the BLS-FALQON algorithm through numerical simulations. The following discussion introduces the definition of the max-cut problem, details the configuration of our simulation methods, and analyzes the numerical results.

\subsection{\label{sec:maxcut} Max-cut problem}

The max-cut problem is an NP-hard problem that serves as a common benchmark in previous study~\cite{magann2022lyapunov,arai2025scalable}. It is widely applied in network design, data clustering, circuit layout, and finance to optimize complex systems~\cite{kernighan1970efficient,barahona1988application,shi2000normalized,ding2001min,glover2022quantum}.

Given an undirected graph $G=(V,E)$ with $n=|V|$ vertices and edge set $E$, the max-cut problem seeks a partition of $V$ into two disjoint subsets $S$ and $\bar{S}$ that maximizes the number of edges crossing the partition. Following the standard quantum encoding~\cite{lucas2014ising}, this objective is mapped to the $n$-qubit problem Hamiltonian:
\begin{equation}
    \label{eq:Hp}
    H_p=-\frac{1}{2}\sum_{(i,j)\in E}(1-Z_iZ_j),
\end{equation}
where $Z_i$ denotes the Pauli-$Z$ operator acting on qubit $i$. The ground state of $H_p$ encodes the optimal partition, and the ground state energy $E_{\text{min}}={\langle H_p\rangle}_{\text{min}}$ equals the negative of the maximum cut value, which is determined via exhaustive calculation.

Finding the exact ground state of $H_p$ is computationally intractable in general, and practical algorithms seek approximate solutions. Solution quality is quantified by the approximation ratio
\begin{equation}
    r_A=\frac{\langle H_p\rangle}{E_{\min}},
\end{equation}
where $r_A=1$ corresponds to the exact optimum. 

The Goemans-Williamson (GW) algorithm~\cite{goemans1995improved} establishes a theoretical classical lower bound approximation ratio for max-cut problem. It achieves $r_{\text{GW}}=0.932$ for unweighted 3-regular graphs. We aim to demonstrate that BLS-FALQON exceeds this threshold with reduced circuit depth and total measurement bases count compared to prior feedback-based methods.

\subsection{\label{sec:method}Simulation settings}

We consider unweighted 3-regular graphs with $n\in\{8,10,12,14,16,18,20\}$ vertices. For $n=8$ and $n=10$, all available non-isomorphic instances (5 and 19, respectively) are included. For each remaining size, 50 non-isomorphic instances are generated at random. The problem Hamiltonian for each instance is constructed according to Eq.~\ref{eq:Hp}. To minimize the total measurement bases count, the Pauli strings are grouped into qubit-wise commuting subsets via the Welsh-Powell algorithm~\cite{welsh1967upper}. The resulting average number of measurement bases per step for each algorithm and graph size is reported in Table~\ref{tab:nbasis}.

\begin{table}[!htb]
\caption{Average number of measurement bases per step required by each algorithm across graph sizes $n\in\{8,10,\dots,20\}$. BLS-FALQON needs to estimate $E_k$ and $A_k$ in a normal step and $E_k$ solely during the triggered backtracking iteration. For $n\ge12$, the per-step measurement bases of BLS-FALQON (normal step) equal those of FALQON, since all $ZZ$ terms of $H_p$ are absorbed into the existing $A_k$ qubit-wise commuting subsets at no extra cost.}
\label{tab:nbasis}
\begin{tabular}{lcccccccc}
\toprule
\multirow{2}{*}{Algorithm} & \multirow{2}{*}{Need to estimate} & \multicolumn{7}{c}{Graph size $n$} \\
\cline{3-9}
 &  & 8 & 10 & 12 & 14 & 16 & 18 & 20 \\
\midrule
FALQON & $A_k$ & 3.2 & 3.05 & 3.4 & 3.32 & 3.46 & 3.42 & 3.42 \\
SO-FALQON & $A_k,B_k,C_k$ & 7.2 & 7.53 & 7.86 & 7.68 & 7.92 & 7.76 & 7.76 \\
BLS-FALQON (in a normal step) & $E_k,A_k$ & 3.4 & 3.16 & 3.4 & 3.32 & 3.46 & 3.42 & 3.42 \\
BLS-FALQON (in backtracking) & $E_k$ & 1 & 1 & 1 & 1 & 1 & 1 & 1 \\
\bottomrule
\end{tabular}
\end{table}

The scaling factor in the backtracking iteration is empirically fixed at $\tau=-1/4$ based on preliminary small-scale tests, chosen such that it yields a moderate geometric decay of the control coefficient per backtracking iteration without requiring an excessive number of trials to reach a valid value. The maximum number of circuit layers is $L=1000$ for all three algorithms. As prior studies have not provided the exact optimal step size for these dimensions, we finely discretize the step size into a candidate set $\{0.02,0.04,\dots,0.20\}$ to select the appropriate $\Delta t$. For FALQON and SO-FALQON, the adopted $\Delta t$ at each graph size is the largest value for which the approximation ratio increases monotonically without oscillation. BLS-FALQON uses the same $\Delta t$ as SO-FALQON at every graph size. The time step values for each algorithm and graph size are listed in Table~\ref{tab:dt}.

\begin{table}[!htb]
\caption{Time step $\Delta t$ adopted by each algorithm in the numerical simulations across graph sizes $n\in\{8,10,\dots,20\}$. BLS-FALQON utilizes the identical step size as SO-FALQON.}
\label{tab:dt}
\begin{tabular}{lccccccc}
\toprule
\multirow{2}{*}{Algorithm} & \multicolumn{7}{c}{Graph size $n$} \\
\cline{2-8}
 & 8 & 10 & 12 & 14 & 16 & 18 & 20 \\
\midrule
FALQON & 0.04 & 0.02 & 0.02 & 0.02 & 0.02 & 0.02 & 0.02 \\
SO-FALQON & 0.16 & 0.14 & 0.14 & 0.12 & 0.12 & 0.10 & 0.10 \\
BLS-FALQON & 0.16 & 0.14 & 0.14 & 0.12 & 0.12 & 0.10 & 0.10 \\
\bottomrule
\end{tabular}
\end{table}

All numerical simulations are implemented using the Qiskit quantum computing framework~\cite{javadi2024quantum} and executed on a Linux platform, equipped with 80 Intel(R) Xeon(R) Silver 4316 CPUs and an NVIDIA GeForce RTX 4090 GPU. Due to GPU memory limitations, \texttt{statevector} method cannot be performed at large graph scales. Therefore, all quantum circuits are implemented using \texttt{tensor\_network} method. Expectation values of all Pauli strings are estimated using the \texttt{BackendEstimatorV2()} primitive provided by Qiskit. Following the default configuration, each measurement basis is executed with $M=1024$ shots.

\subsection{\label{sec:result}Numerical results and analysis}

To quantitatively assess algorithm performance, we utilize two metrics. The first is the total measurement bases count $N_{\text{bases}}$, defined as the total number of measurement bases required for the approximation ratio to reach $r_{\text{GW}}=0.932$. For FALQON and SO-FALQON, $N_{\text{bases}}$ is the product of the step count and the per-step measurement bases count. For BLS-FALQON, the measurement bases consumed by each backtracking iteration are additionally included.

The second metric is $N_{\text{layers}}$, defined as the number of layers required to reach $r_{\text{GW}}$. Since the physical gate depth depends on the transpilation strategy and the connectivity topology of the target superconducting processor, we use $N_{\text{layers}}$ as a hardware-independent representation for circuit depth throughout this section. The physical gate depth is analyzed separately in Sec.~\ref{sec:real} in the context of real quantum hardware experiments. Both metrics are evaluated as averages over all instances at each graph size $n$. The results for $N_{\text{bases}}$ and $N_{\text{layers}}$ are presented in Fig.~\ref{fig:combined}(a) and Fig.~\ref{fig:combined}(b), respectively.

\begin{figure}
    \centering
    \includegraphics[width=.9\textwidth]{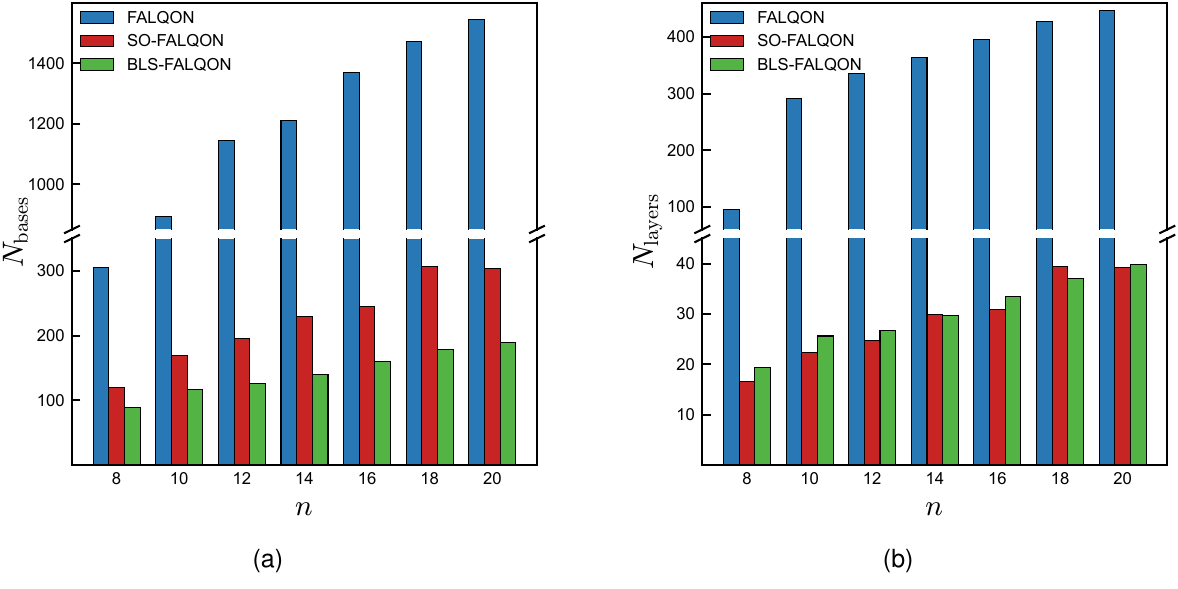}
    \caption{\label{fig:combined}
    Average performance metrics as a function of graph size $n$ for FALQON, SO-FALQON, and BLS-FALQON applied to unweighted 3-regular graphs, all required to reach the approximation ratio $r_{\text{GW}}=0.932$. (a) Total measurement bases count $N_{\text{bases}}$. For BLS-FALQON, the measurement bases consumed during backtracking iterations are included in $N_{\text{bases}}$. (b) Number of circuit layers $N_{\text{layers}}$. The vertical axis breaks indicate truncation for better visibility.
    }
\end{figure}

As shown in Fig.~\ref{fig:combined}(a), BLS-FALQON reduces $N_{\text{bases}}$ by approximately 37.7\% relative to SO-FALQON and 88.1\% relative to FALQON. The reduction over SO-FALQON originates directly from eliminating the per-layer measurements to estimate $B_k$ and $C_k$, as reported in Table~\ref{tab:nbasis}. The reduction in $N_{\text{bases}}$ remains consistent across all graph sizes, indicating that the measurement savings of BLS-FALQON are preserved as the problem scale increases.

Admittedly, Fig.~\ref{fig:combined}(b) shows that BLS-FALQON requires on average 2.5\% more circuit layers than SO-FALQON to reach $r_{\text{GW}}$, while achieving a reduction of approximately 11 $\times$ relative to FALQON. This phenomenon of the increase over SO-FALQON is consistent with the theoretical analysis in Proposition~\ref{prop:2}. The backtracking mechanism identifies a valid coefficient $\beta^{(m^*)}$ that satisfies the descent condition but does not coincide with the second-order optimal $\beta^{\text{(SO)}}$, may result in a marginally smaller energy decrease per layer, thereby slightly increasing the number of layers. Nevertheless, the average overhead of 2.5\% in $N_{\text{layers}}$ is negligible in the context of the 37.7\% reduction in $N_{\text{bases}}$. The circuit depth of BLS-FALQON is therefore considered comparable to that of SO-FALQON. Furthermore, the overhead in $N_{\text{layers}}$ does not grow appreciably with graph size, confirming that BLS-FALQON maintains comparable circuit depth scalability at larger problem scales.

Consequently, the numerical results demonstrate that BLS-FALQON achieves a substantial reduction in total measurement overhead relative to both FALQON and SO-FALQON, with circuit depth comparable to SO-FALQON and significantly lower than FALQON, while maintaining consistent, scalable convergence behavior.

\section{\label{sec:real}Real Quantum Hardware Experiments}

In this section, we validate BLS-FALQON on real quantum hardware. All experiments are executed on the \textit{Tianyan-176} quantum computer, which uses the \textit{zuchongzhi2} superconducting quantum processor~\cite{wu2021strong}. We benchmark BLS-FALQON against SO-FALQON and FALQON across $n\in\{4,6,8\}$ max-cut instances under realistic noise conditions. Sec.~\ref{sec:he} details the hardware specifications, and Sec.~\ref{sec:eraa} presents the experimental setup and comparative results.

\subsection{\label{sec:he}Hardware specifications}

The \textit{zuchongzhi2} quantum processor comprises 66 superconducting qubits arranged in a two-dimensional array with 110 tunable couplers. The qubit connectivity topology and the per-qubit and per-coupler calibration data are shown in Fig.~\ref{fig:topo}. The median calibration parameters during the experimental period are summarized in Table~\ref{tab:tianyan}. All quantum circuits are compiled to OpenQASM 2.0~\cite{cross2017open} prior to execution, and the processor automatically selects the physical qubits with the most favorable calibration metrics for each submitted circuit. For the three graph sizes $n\in\{4,6,8\}$, the physical gate depth per circuit layer after compilation is approximately 54, 65, and 87, respectively.

\begin{figure}
    \centering
    \includegraphics[width=.9\textwidth]{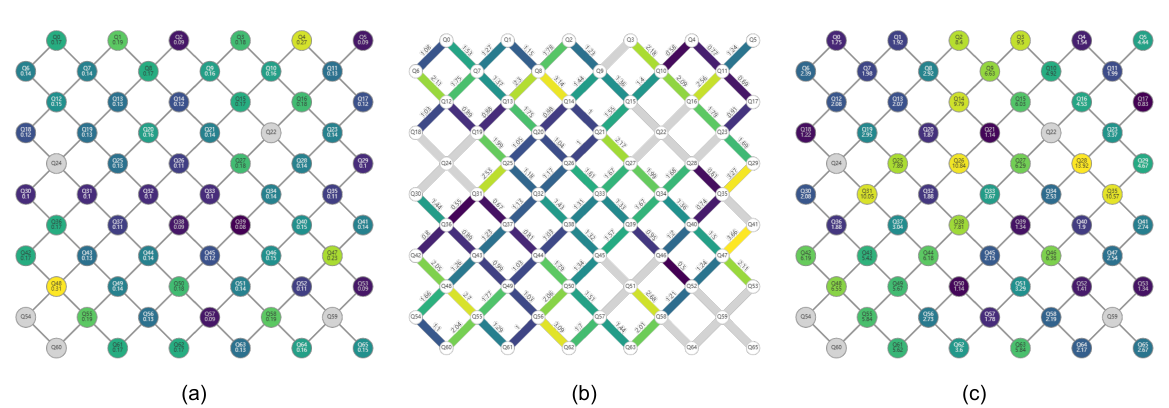}
    \caption{\label{fig:topo}
    Calibration data for the \textit{zuchongzhi2} superconducting quantum processor, retrieved on May 27, 2026. The node-and-edge layout in each panel reflects the qubit connectivity topology of the processor. (a) Single-qubit gate error rate per qubit. (b) Two-qubit CZ gate error rate per coupler. (c) Readout error rate per qubit (unit: \%).
    }
\end{figure}

\begin{table}[!htb]
\caption{Median calibration data for the \textit{zuchongzhi2} superconducting quantum processor, representative of the device performance during the experimental period.}
\label{tab:tianyan}
\begin{tabular}{l c}
\toprule
\textrm{Parameter} & \textrm{Value} \\
\midrule
Experimental period & May 1 -- May 27, 2026 \\
Basis gates & $R_Z(\theta),R_Y(\frac{\pi}{2}),R_Y(-\frac{\pi}{2}),CZ$ \\
Median T1 & 27.89 $\mu$s \\
Median T2 (CPMG) & 21.61 $\mu$s \\
Median 1Q error & 1.4$\times10^{-3}$ \\
Median 2Q error & 1.36$\times10^{-2}$ \\
Median readout error & 2.95$\times10^{-2}$ \\
\bottomrule
\end{tabular}
\end{table}

\subsection{\label{sec:eraa}Experimental results and analysis}

Due to the limited coherence time of the \textit{zuchongzhi2} processor, we restrict the experiments to small-scale max-cut instances. For each graph size $n\in\{4,6,8\}$, one unweighted 3-regular graph instance is selected, whose structure is illustrated in Fig.~\ref{fig:graph}. Following the same step-size selection principle as Sec.~\ref{sec:method}, the adopted $\Delta t$ for each algorithm and graph size is listed in Table~\ref{tab:dtt}. To mitigate stochasticity between runs, each of the three algorithms is executed 20 times on real hardware with a maximum of $L=10$ circuit layers. The scaling factor in hardware experiments is also fixed at $\tau=-1/4$. To isolate the effect of hardware noise on the feedback law from its effect on state preparation, we adopt a hybrid evaluation protocol. Specifically, all expectation values required to determine $\beta_k$ at each layer, including $E_k$ for the backtracking mechanism in BLS-FALQON, are estimated on real quantum hardware, while the approximation ratio $r_A$ is evaluated by substituting the $\beta_k$ sequences obtained from hardware measurements into an exact statevector simulation. This protocol quantifies directly how faithfully the feedback law, as realized under hardware noise, reproduces the ideal control trajectory.

\begin{figure}
    \centering
    \includegraphics[width=.5\textwidth]{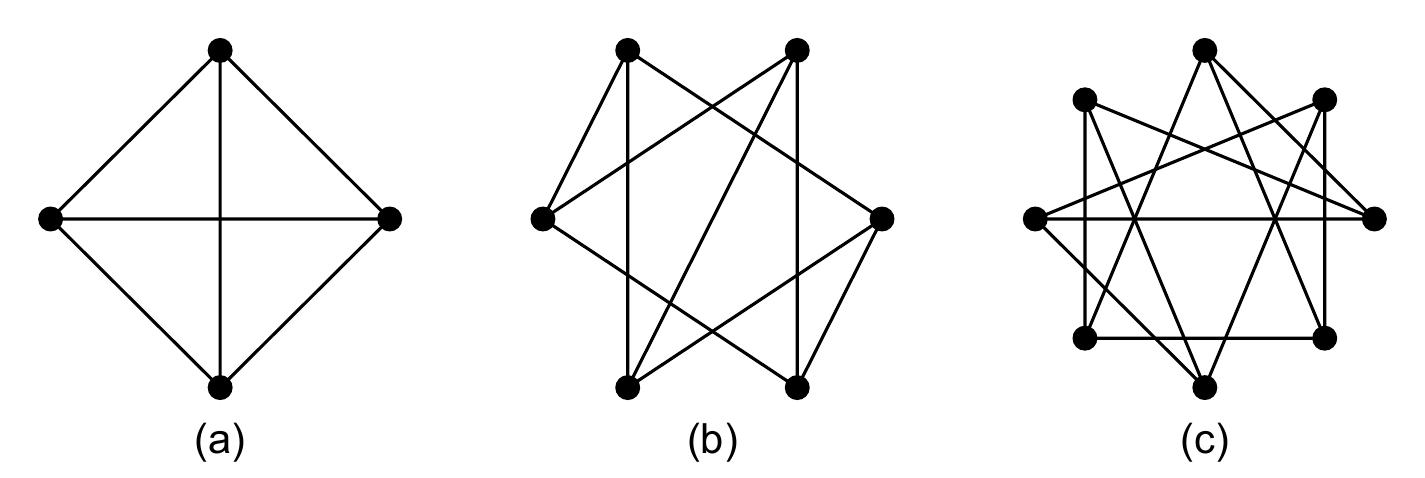}
    \caption{\label{fig:graph}
    Unweighted 3-regular graph instances used in the real hardware experiments for (a) $n=4$, (b) $n=6$, and (c) $n=8$.
    }
\end{figure}

\begin{table}[!htb]
\caption{Time step $\Delta t$ adopted by each algorithm in the real hardware experiments across graph sizes $n\in\{4,6,8\}$. BLS-FALQON utilizes the identical step size as SO-FALQON.}
\label{tab:dtt}
\begin{tabular}{lccc}
\toprule
\textrm{Algorithm} & \textrm{$n=4$} & \textrm{$n=6$} & \textrm{$n=8$} \\
\midrule
FALQON & 0.10 & 0.08 & 0.04 \\
SO-FALQON & 0.20 & 0.16 & 0.16 \\
BLS-FALQON & 0.20 & 0.16 & 0.16 \\
\bottomrule
\end{tabular}
\end{table}

The estimated hardware execution times are summarized in Table~\ref{tab:last}. Following the existing calculation method in Ref.~\cite{guerreschi2019qaoa}, the total time for a single circuit repetition ($T_{run}$) is divided into three sequential parts: initial state preparation ($T_P$), quantum circuit execution ($T_C$), and qubit measurement ($T_M$). The total execution time is formulated as follows:
\begin{equation}
    T_{run}=T_P+T_C+T_M,
\end{equation}
where $T_C=L\cdot d\cdot T_G$.

Here, $L$ denotes the number of circuit layers, $d$ represents the physical gate depth per layer, and $T_G$ is the execution time of a single quantum gate. According to Ref.~\cite{guerreschi2019qaoa}, the combined state preparation and measurement time is $T_P+T_M\approx1$ $\mu$s, while the individual depth time is $T_G\approx10$ ns. Furthermore, the number of shots for each Hamiltonian measurement on the physical hardware is uniformly set to 1024.

\begin{table}[!htb]
    \caption{Average estimated total hardware execution time over $L = 10$ circuit layers across 20 independent hardware runs.}
    \label{tab:last}
    \begin{tabular}{lccc}
    \toprule
    \textrm{Algorithm} & \textrm{$n=4$} & \textrm{$n=6$} & \textrm{$n=8$} \\
    \midrule
    FALQON & 162.61 ms & 140.54 ms & 177.72 ms \\
    SO-FALQON & 365.88 ms & 327.94 ms & 414.67 ms \\
    BLS-FALQON & 206.11 ms & 188.39 ms & 237.35 ms \\
    \bottomrule
    \end{tabular}
\end{table}
    
As observed in Table~\ref{tab:last}, the execution time for $n=6$ is shorter than that for $n=4$. This is because, after applying the Welsh-Powell algorithm for measurement grouping, evaluating $A_k$ in FALQON and $A_k,E_k$ in BLS-FALQON requires measuring only 3 Pauli strings for the $n=6$ instances, whereas it requires 4 strings for $n=4$. Similarly, the grouped observables $A_k,B_k,C_k$ for SO-FALQON necessitate measuring 7 Pauli strings for $n=6$, compared to 9 strings for $n=4$. Overall, the results indicate that BLS-FALQON requires less execution time than SO-FALQON. This reduction is directly attributable to the decreased number of required measurement bases. Ultimately, this experimentally verifies that measurement efficiency can effectively lower the overall time complexity of the algorithm on real quantum hardware.

The approximation ratio results are presented in Fig.~\ref{fig:real}. In each panel, the black curve shows the ideal statevector simulation, the colored curve shows the average $r_A$ over 20 independent hardware runs, and the shaded band indicates the corresponding range across all runs. BLS-FALQON attains the highest average approximation ratio among the three algorithms at all graph sizes and exhibits the narrowest shaded band at $n\in\{4,6\}$, indicating a stable convergence behavior under realistic hardware noise at small-scale instances.

\begin{figure}
    \centering
    \includegraphics[width=.9\textwidth]{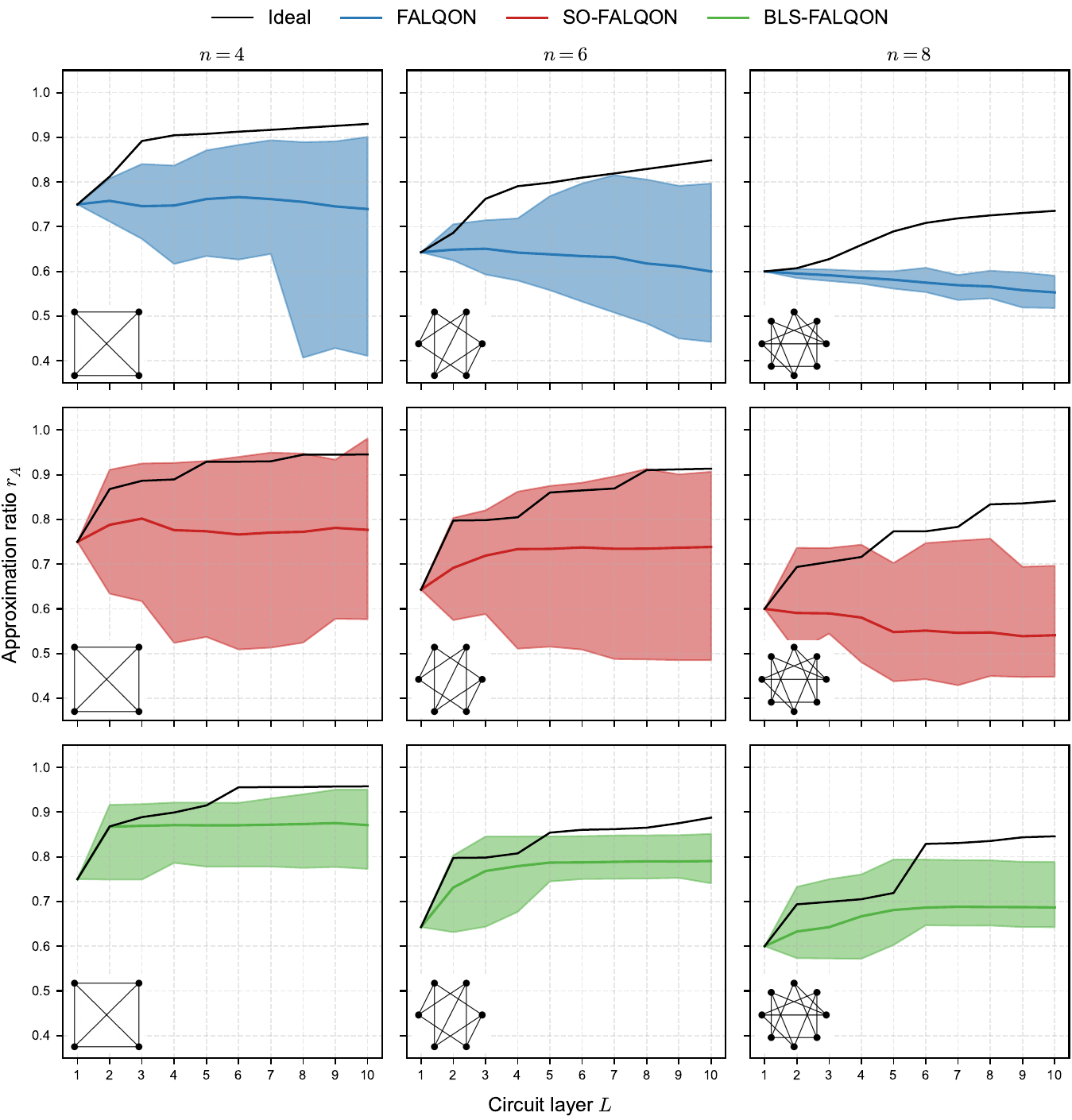}
    \caption{\label{fig:real}
    The approximation ratio results on the \textit{zuchongzhi2} superconducting quantum processor for unweighted 3-regular graph instances over $L=10$ circuit layers. Rows correspond to FALQON, SO-FALQON and BLS-FALQON from top to bottom, and columns correspond to $n=4,6,8$ from left to right. Black curves show the ideal statevector simulation. Colored curves show the average approximation ratio $r_A$ over 20 independent runs. Shaded bands indicate the range of $r_A$ across all runs.
    }
\end{figure}

FALQON operates at smaller time steps. Its average $r_A$ remains near the initial value at $n=4$ and $n=6$, and even decreases with circuit depth at $n=8$. The shaded band of FALQON narrows as $n$ increases, which is a direct consequence of the smaller time step adopted at larger graph sizes. However, this reduced step size does not translate into improved convergence under hardware noise. It merely confines the feedback coefficient $\beta_k^{\text{(FO)}}$ to a smaller range, suppressing improved convergence.

SO-FALQON uses the same large time step as BLS-FALQON but requires estimating $B_k$ and $C_k$ additionally at every layer. The resulting increase in the number of measurement bases amplifies the hardware noise in $\beta_k^{\text{(SO)}}$ at each step, which manifests as the widest shaded band among the three algorithms.

By contrast, BLS-FALQON only requires estimating $A_k$ and $E_k$ in each normal step, reducing the hardware noise contribution to $\beta_k^{\text{(FO)}}$. Furthermore, the backtracking mechanism provides a corrective effect beyond its original role of suppressing oscillations induced by Trotter decomposition. When hardware noise causes a spurious increase in $E_k$, the alternating sign scaling of $\beta_k^{\text{(FO)}}$ identifies a valid descent coefficient without requiring any additional Hamiltonian measurements, partially compensating for misdirection of the control direction caused by hardware noise.

Admittedly, we note that even BLS-FALQON exhibits a visible gap between the hardware runs and the ideal simulation, and this gap widens with increasing problem size. Across all graph sizes, the average $r_A$ shows little improvement beyond layer 5, indicating that noise accumulation over successive layers eventually saturates the useful feedback signal before achieving $r_{\text{GW}}$. These observations reflect the fundamental constraints of current NISQ devices rather than algorithmic limitations. We expect that improvements in qubit coherence times and gate fidelities will substantially narrow this gap and allow the full benefit of BLS-FALQON to be realized at larger problem scales.

\section{\label{sec:conclu}Conclusions}

In this work, we propose BLS-FALQON to achieve both low circuit depth and reduced measurement overhead compared to FALQON, as our method incorporates the backtracking line search strategy to dynamically adjust the control coefficient. Numerical simulations prove that BLS-FALQON successfully outperforms FALQON in both circuit depth and measurement efficiency, while it further reduces total measurement count by 37.7\% relative to SO-FALQON and maintaining a comparable circuit depth. Real quantum hardware experiments demonstrate that our proposed method decreases the practical execution time on NISQ devices, and remains robust and effective under realistic hardware noise at small problem scales.

For the joint estimation of multiple low-locality Pauli strings, an alternative to Welsh-Powell grouping is the classical shadows protocol~\cite{james2001measurement,aaronson2018shadow,huang2020predicting}, as demonstrated in a prior work~\cite{bertuzzi2025shadow}. This does not affect the contribution of our paper, since our contribution lies in reducing the number of Hamiltonians that need to be estimated, independent of the specific approach used to estimate each of them. In this paper, we adopt the same Welsh-Powell grouping as SO-FALQON to implement this reduction, but the measurement scheme could instead be implemented using classical shadows, and the conclusion of this paper would remain consistent with either choice. Whether Welsh-Powell grouping or classical shadows offers a greater practical advantage in solving combinatorial optimization problems is left as a direction for future work.

\section{Acknowledgements}

This work is supported by National Natural Science Foundation of China (Grant Nos. U25B2014, 62371069, 62372048, 62272056,), and the National Key Laboratory of Security Communication Foundation (2025, 6142103042503). This work is partly supported by Tianyan Quantum Computing Cloud Platform developed by China Telecom Quantum Information Technology Group Co., Ltd.

\printcredits

\bibliographystyle{elsarticle-num-names}

\bibliography{bls_falqon_huang}

@article{magann2022feedback,
  title={Feedback-based quantum optimization},
  author={Magann, Alicia B and Rudinger, Kenneth M and Grace, Matthew D and Sarovar, Mohan},
  journal={Phys. Rev. Lett.},
  volume={129},
  number={25},
  pages={250502},
  year={2022},
  publisher={APS},
  doi={https://doi.org/10.1103/PhysRevLett.129.250502}
}

@article{magann2022lyapunov,
  title={Lyapunov-control-inspired strategies for quantum combinatorial optimization},
  author={Magann, Alicia B and Rudinger, Kenneth M and Grace, Matthew D and Sarovar, Mohan},
  journal={Phys. Rev. A},
  volume={106},
  number={6},
  pages={062414},
  year={2022},
  publisher={APS},
  doi={https://doi.org/10.1103/PhysRevA.106.062414}
}

@article{kosloff1992excitation,
  title={Excitation without demolition: Radiative excitation of ground-surface vibration by impulsive stimulated Raman scattering with damage control},
  author={Kosloff, Ronnie and Hammerich, Audrey Dell and Tannor, David},
  journal={Phys. Rev. Lett.},
  volume={69},
  number={15},
  pages={2172},
  year={1992},
  publisher={APS},
  doi={https://doi.org/10.1103/PhysRevLett.69.2172}
}

@article{ohtsuki1998application,
  title={Application of a locally optimized control theory to pump--dump laser-driven chemical reactions},
  author={Ohtsuki, Y and Yahata, Y and Kono, H and Fujimura, Y},
  journal={Chem. Phys. Lett.},
  volume={287},
  number={5-6},
  pages={627--631},
  year={1998},
  publisher={Elsevier},
  doi={https://doi.org/10.1016/S0009-2614(98)00224-3}
}

@article{j1999laser,
  title={Laser cooling of internal degrees of freedom of molecules by dynamically trapped states},
  author={J. Tannor, David and Kosloff, Ronnie and Bartana, Alon},
  journal={Faraday Discuss.},
  volume={113},
  pages={365--383},
  year={1999},
  publisher={The Royal Society of Chemistry},
  doi={https://doi.org/10.1039/a902103e}
}

@inproceedings{grivopoulos2003lyapunov,
  title={Lyapunov-based control of quantum systems},
  author={Grivopoulos, Symeon and Bamieh, Bassam},
  booktitle={42nd IEEE International Conference on Decision and Control (IEEE Cat. No. 03CH37475)},
  volume={1},
  pages={434--438},
  year={2003},
  organization={IEEE},
  doi={https://doi.org/10.1109/CDC.2003.1272601}
}

@article{engel2009local,
  title={Local control theory: Recent applications to energy and particle transfer processes in molecules},
  author={Engel, Volker and Meier, Christoph and Tannor, David J},
  journal={Adv. Chem. Phys.},
  volume={141},
  pages={29},
  year={2009},
  doi={https://doi.org/10.1002/9780470431917.ch2}
}

@article{peruzzo2014variational,
  title={A variational eigenvalue solver on a photonic quantum processor},
  author={Peruzzo, Alberto and McClean, Jarrod and Shadbolt, Peter and Yung, Man-Hong and Zhou, Xiao-Qi and Love, Peter J and Aspuru-Guzik, Al{\'a}n and O’brien, Jeremy L},
  journal={Nat. Commun.},
  volume={5},
  number={1},
  pages={4213},
  year={2014},
  publisher={Nature Publishing Group UK London},
  doi={https://doi.org/10.1038/ncomms5213}
}

@article{farhi2014quantum,
  title={A quantum approximate optimization algorithm},
  author={Farhi, Edward and Goldstone, Jeffrey and Gutmann, Sam},
  journal={arXiv preprint arXiv:1411.4028},
  year={2014},
  doi={https://doi.org/10.48550/arXiv.1411.4028}
}

@article{liu2021variational,
  title={Variational quantum algorithm for the Poisson equation},
  author={Liu, Hai-Ling and Wu, Yu-Sen and Wan, Lin-Chun and Pan, Shi-Jie and Qin, Su-Juan and Gao, Fei and Wen, Qiao-Yan},
  journal={Phys. Rev. A},
  volume={104},
  number={2},
  pages={022418},
  year={2021},
  publisher={APS},
  doi={https://doi.org/10.1103/PhysRevA.104.022418}
}

@article{bravo2023variational,
  title={Variational quantum linear solver},
  author={Bravo-Prieto, Carlos and LaRose, Ryan and Cerezo, Marco and Subasi, Yigit and Cincio, Lukasz and Coles, Patrick J},
  journal={Quantum},
  volume={7},
  pages={1188},
  year={2023},
  publisher={Verein zur F{\"o}rderung des Open Access Publizierens in den Quantenwissenschaften},
  doi={https://doi.org/10.22331/q-2023-11-22-1188}
}

@article{cerezo2021variational,
  title={Variational quantum algorithms},
  author={Cerezo, Marco and Arrasmith, Andrew and Babbush, Ryan and Benjamin, Simon C and Endo, Suguru and Fujii, Keisuke and McClean, Jarrod R and Mitarai, Kosuke and Yuan, Xiao and Cincio, Lukasz and others},
  journal={Nat. Rev. Phys.},
  volume={3},
  number={9},
  pages={625--644},
  year={2021},
  publisher={Nature Publishing Group UK London},
  doi={https://doi.org/10.1038/s42254-021-00348-9}
}

@article{bittel2021training,
  title={Training variational quantum algorithms is NP-hard},
  author={Bittel, Lennart and Kliesch, Martin},
  journal={Phys. Rev. Lett.},
  volume={127},
  number={12},
  pages={120502},
  year={2021},
  publisher={APS},
  doi={https://doi.org/10.1103/PhysRevLett.127.120502}
}

@article{zhou2020quantum,
  title={Quantum approximate optimization algorithm: Performance, mechanism, and implementation on near-term devices},
  author={Zhou, Leo and Wang, Sheng-Tao and Choi, Soonwon and Pichler, Hannes and Lukin, Mikhail D},
  journal={Phys. Rev. X},
  volume={10},
  number={2},
  pages={021067},
  year={2020},
  publisher={APS},
  doi={https://doi.org/10.1103/PhysRevX.10.021067}
}

@article{harrigan2021quantum,
  title={Quantum approximate optimization of non-planar graph problems on a planar superconducting processor},
  author={Harrigan, Matthew P and Sung, Kevin J and Neeley, Matthew and Satzinger, Kevin J and Arute, Frank and Arya, Kunal and Atalaya, Juan and Bardin, Joseph C and Barends, Rami and Boixo, Sergio and others},
  journal={Nat. Phys.},
  volume={17},
  number={3},
  pages={332--336},
  year={2021},
  publisher={Nature Publishing Group UK London},
  doi={https://doi.org/10.1038/s41567-020-01105-y}
}

@article{kruger2025out,
  title={Out of the Loop: Structural Approximation of Optimisation Landscapes and non-Iterative Quantum Optimisation},
  author={Kr{\"u}ger, Tom and Mauerer, Wolfgang},
  journal={Quantum},
  volume={9},
  pages={1903},
  year={2025},
  publisher={Verein zur F{\"o}rderung des Open Access Publizierens in den Quantenwissenschaften},
  doi={https://doi.org/10.22331/q-2025-11-06-1903}
}

@article{regev2025efficient,
  title={An efficient quantum factoring algorithm},
  author={Regev, Oded},
  journal={J. ACM},
  volume={72},
  number={1},
  pages={1--13},
  year={2025},
  publisher={ACM New York, NY},
  doi={https://doi.org/10.1145/3708471}
}

@inproceedings{ragavan2024space,
  title={Space-efficient and noise-robust quantum factoring},
  author={Ragavan, Seyoon and Vaikuntanathan, Vinod},
  booktitle={Annual International Cryptology Conference},
  pages={107--140},
  year={2024},
  organization={Springer},
  doi={https://doi.org/10.1007/978-3-031-68391-6_4}
}

@article{xiao2026offline,
  title={Offline Dedicated Quantum Attacks on Block Cipher Constructions Based on Two Parallel Permutation-Based Pseudorandom Functions},
  author={Zhen, XiaoFan and Li, Zhenqiang and Fan, Jiacheng and Qin, Sujuan and Gao, Fei},
  journal={Sci. China Inf. Sci.},
  volume={69},
  number={8},
  pages={180507},
  year={2026},
  publisher={Science China Press},
  doi={https://doi.org/10.1007/s11432-026-5030-8}
}

@article{preskill2018quantum,
  title={Quantum computing in the NISQ era and beyond},
  author={Preskill, John},
  journal={Quantum},
  volume={2},
  pages={79},
  year={2018},
  publisher={Verein zur F{\"o}rderung des Open Access Publizierens in den Quantenwissenschaften},
  doi={https://doi.org/10.22331/q-2018-08-06-79}
}

@article{li2025efficient,
  title={An efficient quantum proactive incremental learning algorithm},
  author={Li, Lingxiao and Li, Jing and Song, Yanqi and Qin, Sujuan and Wen, Qiaoyan and Gao, Fei},
  journal={Sci. China Phys. Mech. Astron.},
  volume={68},
  number={1},
  pages={210313},
  year={2025},
  publisher={Springer},
  doi={https://doi.org/10.1007/s11433-024-2501-4}
}

@article{su2025topology,
  title={Topology-driven quantum architecture search framework},
  author={Su, Junjian and Fan, Jiacheng and Wu, Shengyao and Li, Guanghui and Qin, Sujuan and Gao, Fei},
  journal={Sci. China Inf. Sci.},
  volume={68},
  number={8},
  pages={180507},
  year={2025},
  publisher={Springer},
  doi={https://doi.org/10.1007/s11432-024-4486-x}
}

@article{arai2025scalable,
  title={Scalable circuit depth reduction in feedback-based quantum optimization with a quadratic approximation},
  author={Arai, Don and Okada, Ken N and Nakano, Yuichiro and Mitarai, Kosuke and Fujii, Keisuke},
  journal={Phys. Rev. Res.},
  volume={7},
  number={1},
  pages={013035},
  year={2025},
  publisher={APS},
  doi={https://doi.org/10.1103/PhysRevResearch.7.013035}
}

@article{malla2024feedback,
  title={Feedback-based quantum algorithm inspired by counterdiabatic driving},
  author={Malla, Rajesh K and Sukeno, Hiroki and Yu, Hongye and Wei, Tzu-Chieh and Weichselbaum, Andreas and Konik, Robert M},
  journal={Phys. Rev. Res.},
  volume={6},
  number={4},
  pages={043068},
  year={2024},
  publisher={APS},
  doi={https://doi.org/10.1103/PhysRevResearch.6.043068}
}

@article{tang2025nonvariational,
  title={Nonvariational ADAPT algorithm for quantum simulations},
  author={Tang, Ho Lun and Chen, Yanzhu and Biswas, Prakriti and Magann, Alicia B and Arenz, Christian and Economou, Sophia E},
  journal={Phys. Rev. Res.},
  volume={7},
  number={2},
  pages={023275},
  year={2025},
  publisher={APS},
  doi={https://doi.org/10.1103/PhysRevResearch.7.023275}
}

@article{brady2025feedback,
  title={Feedback-based optimally controlled quantum states},
  author={Brady, Lucas T and Hadfield, Stuart},
  journal={Phys. Rev. A},
  volume={111},
  number={6},
  pages={062406},
  year={2025},
  publisher={APS},
  doi={https://doi.org/10.1103/PhysRevA.111.062406}
}

@article{rattighieri2025accelerating,
  title={Accelerating feedback-based quantum algorithms through time rescaling},
  author={Rattighieri, LAM and Pexe, GEL and Bernardo, BL and Fanchini, FF},
  journal={Phys. Rev. A},
  volume={112},
  number={4},
  pages={042607},
  year={2025},
  publisher={APS},
  doi={https://doi.org/10.1103/PhysRevA.112.042607}
}

@article{bertuzzi2025shadow,
  title={Shadow measurements for feedback-based quantum optimization},
  author={Bertuzzi, Let{\'\i}cia and Engster, Jo{\~a}o P and da Rosa, Evandro CR and Duzzioni, Eduardo I},
  journal={Physical Review A},
  volume={112},
  number={2},
  pages={022419},
  year={2025},
  publisher={APS},
  doi={https://doi.org/10.1103/PhysRevA.112.022419}
}

@article{james2001measurement,
  title={Measurement of qubits},
  author={James, Daniel FV and Kwiat, Paul G and Munro, William J and White, Andrew G},
  journal={Phys. Rev. A},
  volume={64},
  number={5},
  pages={052312},
  year={2001},
  publisher={APS},
  doi={https://doi.org/10.1103/PhysRevA.64.052312}
}

@inproceedings{aaronson2018shadow,
  title={Shadow tomography of quantum states},
  author={Aaronson, Scott},
  booktitle={Proceedings of the 50th annual ACM SIGACT symposium on theory of computing},
  pages={325--338},
  year={2018},
  doi={https://doi.org/10.1145/3188745.3188802}
}

@article{huang2020predicting,
  title={Predicting many properties of a quantum system from very few measurements},
  author={Huang, Hsin-Yuan and Kueng, Richard and Preskill, John},
  journal={Nat. Phys.},
  volume={16},
  number={10},
  pages={1050--1057},
  year={2020},
  publisher={Nature Publishing Group UK London},
  doi={https://doi.org/10.1038/s41567-020-0932-7}
}

@article{armijo1966minimization,
  title={Minimization of functions having Lipschitz continuous first partial derivatives},
  author={Armijo, Larry},
  journal={Pac. J. Math.},
  volume={16},
  number={1},
  pages={1--3},
  year={1966},
  publisher={Mathematical Sciences Publishers},
  doi={https://doi.org/10.2140/pjm.1966.16.1}
}

@article{bertsekas1997nonlinear,
  title={Nonlinear programming},
  author={Bertsekas, Dimitri P},
  journal={J. Oper. Res. Soc.},
  volume={48},
  number={3},
  pages={334--334},
  year={1997},
  publisher={Taylor \& Francis},
  doi={https://doi.org/10.1057/palgrave.jors.2600425}
}

@book{nocedal2006numerical,
  title={Numerical optimization},
  author={Nocedal, Jorge and Wright, Stephen J},
  year={2006},
  publisher={Springer},
  doi={https://doi.org/10.1007/978-0-387-40065-5}
}

@article{ni2024multilevel,
  title={Multilevel leapfrogging initialization strategy for quantum approximate optimization algorithm},
  author={Ni, Xiao-Hui and Cai, Bin-Bin and Liu, Hai-Ling and Qin, Su-Juan and Gao, Fei and Wen, Qiao-Yan},
  journal={Adv. Quantum Technol.},
  volume={7},
  number={5},
  pages={2300419},
  year={2024},
  publisher={Wiley Online Library},
  doi={https://doi.org/10.1002/qute.202300419}
}

@article{ni2025progressive,
  title={Progressive quantum algorithm for maximum independent set with quantum alternating operator ansatz},
  author={Ni, Xiao-Hui and Li, Ling-Xiao and Song, Yan-Qi and Jin, Zheng-Ping and Qin, Su-Juan and Gao, Fei},
  journal={Chin. Phys. B},
  volume={34},
  number={7},
  pages={070304},
  year={2025},
  publisher={Chinese Physical Society and IOP Publishing Ltd},
  doi={https://doi.org/10.1088/1674-1056/addd83}
}

@article{ni2026adaptive,
  title={An Adaptive Mixer Allocation Strategy for the Quantum Alternating Operator Ansatz},
  author={Ni, Xiao-Hui and Wu, Yu-Sen and Cai, Bin-Bin and Li, Wen-Min and Qin, Su-Juan and Gao, Fei},
  journal={Adv. Quantum Technol.},
  volume={9},
  number={3},
  pages={e00487},
  year={2026},
  publisher={Wiley Online Library},
  doi={https://doi.org/10.1002/qute.202500487}
}

@article{huang2026iterative,
  title={Iterative partition-search variational quantum algorithm for solving the shortest-vector problem},
  author={Huang, Zi-Wen and Ni, Xiao-Hui and Fan, Jia-Cheng and Qin, Su-Juan and Huang, Wei and Xu, Bing-Jie and Gao, Fei},
  journal={Phys. Rev. A},
  volume={113},
  number={3},
  pages={032601},
  year={2026},
  publisher={APS},
  doi={https://doi.org/10.1103/PhysRevA.113.032601}
}

@article{wu2021strong,
  title={Strong quantum computational advantage using a superconducting quantum processor},
  author={Wu, Yulin and Bao, Wan-Su and Cao, Sirui and Chen, Fusheng and Chen, Ming-Cheng and Chen, Xiawei and Chung, Tung-Hsun and Deng, Hui and Du, Yajie and Fan, Daojin and others},
  journal={Phys. Rev. Lett.},
  volume={127},
  number={18},
  pages={180501},
  year={2021},
  publisher={APS},
  doi={https://doi.org/10.1103/PhysRevLett.127.180501}
}

@article{zhu2022quantum,
  title={Quantum computational advantage via 60-qubit 24-cycle random circuit sampling},
  author={Zhu, Qingling and Cao, Sirui and Chen, Fusheng and Chen, Ming-Cheng and Chen, Xiawei and Chung, Tung-Hsun and Deng, Hui and Du, Yajie and Fan, Daojin and Gong, Ming and others},
  journal={Sci. Bull.},
  volume={67},
  number={3},
  pages={240--245},
  year={2022},
  publisher={Elsevier},
  doi={https://doi.org/10.1016/j.scib.2021.10.017}
}

@article{cao2023generation,
  title={Generation of genuine entanglement up to 51 superconducting qubits},
  author={Cao, Sirui and Wu, Bujiao and Chen, Fusheng and Gong, Ming and Wu, Yulin and Ye, Yangsen and Zha, Chen and Qian, Haoran and Ying, Chong and Guo, Shaojun and others},
  journal={Nature},
  volume={619},
  number={7971},
  pages={738--742},
  year={2023},
  publisher={Nature Publishing Group UK London},
  doi={https://doi.org/10.1038/s41586-023-06195-1}
}

@article{qian2025programmable,
  title={Programmable higher-order nonequilibrium topological phases on a superconducting quantum processor},
  author={Qian, Haoran and Gong, Ming and Zhang, Jiahui and Guo, Shaojun and Zha, Chen and Chen, Fusheng and Ye, Yangsen and Wu, Yulin and Cao, Sirui and Ying, Chong and others},
  journal={Science},
  volume={390},
  number={6776},
  pages={930--934},
  year={2025},
  publisher={American Association for the Advancement of Science},
  doi={https://doi.org/10.1126/science.adp6802}
}

@article{welsh1967upper,
  title={An upper bound for the chromatic number of a graph and its application to timetabling problems},
  author={Welsh, Dominic JA and Powell, Martin B},
  journal={Comput. J.},
  volume={10},
  number={1},
  pages={85--86},
  year={1967},
  publisher={Oxford University Press},
  doi={https://doi.org/10.1093/comjnl/10.1.85}
}

@book{rossmann2006lie,
  title={Lie groups: an introduction through linear groups},
  author={Rossmann, Wulf},
  volume={5},
  year={2006},
  publisher={OUP Oxford},
  doi={https://doi.org/10.1093/oso/9780198596837.001.0001}
}

@book{brian2003lie,
  title={Lie groups, Lie algebras, and representations: an elementary introduction},
  author={Brian, C Hall},
  year={2003},
  publisher={Springer},
  doi={https://doi.org/10.1007/978-0-387-21554-9}
}

@article{kernighan1970efficient,
  title={An efficient heuristic procedure for partitioning graphs},
  author={Kernighan, Brian W and Lin, Shen},
  journal={Bell Syst. Tech. J.},
  volume={49},
  number={2},
  pages={291--307},
  year={1970},
  publisher={Nokia Bell Labs},
  doi={https://doi.org/10.1002/j.1538-7305.1970.tb01770.x}
}

@article{barahona1988application,
  title={An application of combinatorial optimization to statistical physics and circuit layout design},
  author={Barahona, Francisco and Gr{\"o}tschel, Martin and J{\"u}nger, Michael and Reinelt, Gerhard},
  journal={Oper. Res.},
  volume={36},
  number={3},
  pages={493--513},
  year={1988},
  publisher={INFORMS},
  doi={https://doi.org/10.1287/opre.36.3.493}
}

@article{shi2000normalized,
  title={Normalized cuts and image segmentation},
  author={Shi, Jianbo and Malik, Jitendra},
  journal={IEEE Trans. Pattern Anal. Mach. Intell.},
  volume={22},
  number={8},
  pages={888--905},
  year={2000},
  publisher={Ieee},
  doi={https://doi.org/10.1109/34.868688}
}

@inproceedings{ding2001min,
  title={A min-max cut algorithm for graph partitioning and data clustering},
  author={Ding, Chris HQ and He, Xiaofeng and Zha, Hongyuan and Gu, Ming and Simon, Horst D},
  booktitle={Proceedings 2001 IEEE international conference on data mining},
  pages={107--114},
  year={2001},
  organization={IEEE},
  doi={https://doi.org/10.1109/ICDM.2001.989507}
}

@article{glover2022quantum,
  title={Quantum bridge analytics I: a tutorial on formulating and using QUBO models},
  author={Glover, Fred and Kochenberger, Gary and Hennig, Rick and Du, Yu},
  journal={Ann. Oper. Res.},
  volume={314},
  number={1},
  pages={141--183},
  year={2022},
  publisher={Springer},
  doi={https://doi.org/10.1007/s10479-022-04634-2}
}

@article{lucas2014ising,
  title={Ising formulations of many NP problems},
  author={Lucas, Andrew},
  journal={Front. Phys.},
  volume={2},
  pages={74887},
  year={2014},
  publisher={Frontiers},
  doi={https://doi.org/10.3389/fphy.2014.00005}
}

@article{goemans1995improved,
  title={Improved approximation algorithms for maximum cut and satisfiability problems using semidefinite programming},
  author={Goemans, Michel X and Williamson, David P},
  journal={J. ACM},
  volume={42},
  number={6},
  pages={1115--1145},
  year={1995},
  publisher={ACM New York, NY, USA},
  doi={https://doi.org/10.1145/227683.227684}
}

@article{javadi2024quantum,
  title={Quantum computing with Qiskit},
  author={Javadi-Abhari, Ali and Treinish, Matthew and Krsulich, Kevin and Wood, Christopher J and Lishman, Jake and Gacon, Julien and Martiel, Simon and Nation, Paul D and Bishop, Lev S and Cross, Andrew W and others},
  journal={arXiv preprint arXiv:2405.08810},
  year={2024},
  doi={https://doi.org/10.48550/arXiv.2405.08810}
}

@article{cross2017open,
  title={Open quantum assembly language},
  author={Cross, Andrew W and Bishop, Lev S and Smolin, John A and Gambetta, Jay M},
  journal={arXiv preprint arXiv:1707.03429},
  year={2017},
  doi={https://doi.org/10.48550/arXiv.1707.03429}
}

@article{guerreschi2019qaoa,
  title={QAOA for Max-Cut requires hundreds of qubits for quantum speed-up},
  author={Guerreschi, Gian Giacomo and Matsuura, Anne Y},
  journal={Sci. Rep.},
  volume={9},
  number={1},
  pages={6903},
  year={2019},
  publisher={Nature Publishing Group UK London},
  doi={https://doi.org/10.1038/s41598-019-43176-9}
}

\end{document}